[1994/12/01]
\documentclass{aomart}
\pdfoutput=1
\usepackage{tikz}
\usetikzlibrary{calc,arrows,automata}
\usetikzlibrary{matrix,positioning,arrows,decorations.pathmorphing,shapes}
\usetikzlibrary{shapes,snakes}

\usepackage{algorithm}
\usepackage{algorithmic}
\chardef\bslash=`\\ 

\newtheorem[{}\it]{thm}{Theorem}[section]
\newtheorem{cor}[thm]{Corollary}

\newtheorem{prop}[thm]{Proposition}

\theoremstyle{definition}
\newtheorem{defn}{Definition}[section]

\newtheorem*[{}\it]{notation}{Notation}

\newcommand{\eval}[2][\right]{\relax
  \ifx#1\right\relax \left.\fi#2#1\rvert}

\title[Irrelevant Components and Exact Computation of the Diameter Constrained Reliability]
{Irrelevant Components and Exact Computation of the Diameter Constrained Reliability}
\author{Eduardo Canale, Pablo Romero, Gerardo Rubino}
\email{canale@fing.edu.uy; pablo.promero@inria.fr; rubino@inria.fr} 
\address{Campus de la Universidad Nacional de Asunci\'on.\\ 
Ruta Mcal. Estigarribia, Km. 10,5. San Lorenzo - Paraguay.\\
Inria Rennes, Bretagne-Atlantique, Campus de Beaulieu.\\
PC 35042, RENNES Cedex}

\subject{primary}{matsc2000}{123456789}
\subject{secondary}{matsc2000}{FFFF}

\volumenumber{1}
\issuenumber{1}
\publicationyear{2014}

\begin{document}
 
\begin{abstract}
Let $G=(V,E)$ be a simple graph with $|V|=n$ nodes and $|E|=m$ links, a subset  
$K \subseteq V$ of \emph{terminals}, a vector $p=(p_1,\ldots,p_m) \in [0,1]^m$ and a positive integer $d$, called \emph{diameter}. 
We assume nodes are perfect but links fail stochastically and independently, with 
probabilities $q_i=1-p_i$. The \emph{diameter-constrained reliability} (DCR for short), is the probability that the 
terminals of the resulting subgraph remain connected by paths composed by $d$ links, or less. 
This number is denoted by $R_{K,G}^{d}(p)$. The general computation of the parameter $R_{K,G}^{d}(p)$ belongs to the class of $\mathcal{N}\mathcal{P}$-Hard problems, 
since is subsumes the complexity that a random graph is connected.\\  

A discussion of the computational complexity for DCR-subproblems is provided in terms of the number of terminal nodes $k=|K|$ and diameter $d$. 
Either when $d=1$ or when $d=2$ and $k$ is fixed, the DCR is inside the class $\mathcal{P}$ of polynomial-time problems.  
The DCR turns $\mathcal{N}\mathcal{P}$-Hard even if $k \geq 2$ and $d\geq 3$ are fixed, or in an all-terminal scenario when $d=2$. 
The traditional approach is to design either exponential exact algorithms or efficient solutions for particular graph classes.\\ 
 
The contributions of this paper are two-fold. First, a new recursive class of graphs are shown to have efficient DCR computation. 
Second, we define a factorization method in order to develop an exact DCR computation in general. 
The approach is inspired in prior works related with the determination of irrelevant links and deletion-contraction formula.  
\end{abstract}

\maketitle
\tableofcontents

\section{Motivation}\label{sect:intro}
The definition of DCR has been introduced in 2001 by H\'ector Cancela and Louis Petingi, inspired in delay-sensitive applications over 
the Internet infrastructure~\cite{PR01}. 
Nevertheless, its applications over other fields of knowledge enriches the motivation 
of this problem in the research community~\cite{Colbourn99reliabilityissues}. 

We wish to communicate special nodes in a network, called \emph{terminals}, 
by $d$ hops or less, in a scenario where nodes are perfect but links fail stochastically and independently. 
The all-terminal case with $d=n-1$ is precisely the probability that a random graph is connected, or 
\emph{classical reliability problem} (CLR for short). Arnon Rosenthal proved that the CLR is inside the class 
of $\mathcal{N}\mathcal{P}$-Hard problems~\cite{Rosenthal}. As a corollary, the general DCR is $\mathcal{N}\mathcal{P}$-Hard as well, 
hence intractable unless $\mathcal{P}=\mathcal{N}\mathcal{P}$. 

The focus of this paper is an exact computation of the DCR in a source-terminal scenario. 
The paper is structured in the following manner. 
In Section~\ref{DCR}, a formal definition of DCR is provided as a particular instance of a coherent stochastic binary system. 
The computational complexity of the DCR is discussed in terms of the diameter and number of terminals in Section~\ref{Prior}. 
There, we can appreciate that the source-terminal DCR problem belongs to the class of $\mathcal{NP}$-Hard problems, whenever $d \geq 3$. 
Therefore, only exponential-time (i.e. over-polynomial) exact algorithms are feasible, unless $\mathcal{P}=\mathcal{NP}$.\\ 

The main contributions are summarized in Sections~\ref{Exact}, \ref{Irrelevant-Components} and~\ref{Algorithm}. 
Specifically, a new family of graphs with efficient source-terminal DCR computation is introduced in Section~\ref{Exact}. 
A discussion on the determination of irrelevant components (links) is provided in Section~\ref{Irrelevant-Components}. 
There, we also include a sufficient condition to have an irrelevant component, which is stronger than previous sufficient conditions. 
Nevertheless, the determination of irrelevant links is still an open problem. In Section~\ref{Algorithm}, 
a set of elementary operations that are DCR-invariants are presented. Additionally, an algorithm combines these 
invariants with the deletion of irrelevant links and a recursive decomposition method in graphs. The spirit is to 
reduce/simplify the network as most as possible during each iteration, in order to return the source-terminal DCR exactly 
and efficiently. Finally, a summary of open problems and trends for future work is included in Section~\ref{Conclusion}.

\section{Terminology}\label{DCR}
We are given a system with $m$ components. These components are either ``up'' or ``down'', and 
the binary state is captured by a binary word $x=(x_1,\ldots,x_m) \in \{0,1\}^m$. Additionally we have a 
structure function $\phi: \{0,1\}^m \to \{0,1\}$ such that $\phi(x)=1$ if the system works 
under state $x$, and $\phi(x)=0$ otherwise. When the components work independently and stochastically 
with certain probabilities of operation $p=(p_1,\ldots,p_m)$, the pair $(\phi,p)$ defines a \emph{stochastic binary system}, 
or SBS for short, following the terminology of Michael Ball~\cite{Ball1986}. 
An SBS is \emph{coherent} whenever $x\leq y$ implies that $\phi(x) \leq \phi(y)$, where the partial order set $(\leq,\{0,1\}^m)$ 
is bit-wise (i.e. $x\leq y$ if and only if $x_i \leq y_i$ for all $i\in \{1,\ldots,m\}$). If $\{X_i\}_{i=1,\ldots,m}$ 
is a set of independent binary random variables with $P(X_i=1)=p_i$ and $X=(X_1,\ldots,X_m)$, then 
$r=E(\phi(X)) = P(\phi(X)=1)$ is the \emph{reliability} of the SBS.

We recall a bit of terminology coming from graph theory, which will be used throughout this treatment. 
A graph $G=(V,E)$ is \emph{bipartite} if there exists a bipartition $V = V_1 \cup V_2$ such that 
$E \subseteq \{ \{x,y\}: x \in V_1, y \in V_2\}$. 
A~\emph{vertex cover} in a graph $G=(V,E)$ is a subset $V^{\prime} \subseteq V$ such that $V^{\prime}$ meets all links in $E$. 
%
If $P=\{V_1,\ldots,V_c\}$ is a partition of $V$, the \emph{quotient graph} is $G^{\prime} = (P,E^{\prime})$, 
where $\{V_i,V_j\} \in E^{\prime}$ if and only if $i\neq j$ and there exists an edge from a vertex of $V_i$ to 
a vertex of $V_j$ in $E$. 
We say $v_j$ is \emph{reachable} from $v_i$ either when $v_i=v_j$ or there is a path from $v_i$ to $v_j$. 
In a simple graph $G$, reachability is an equivalence relation, and $c$, the number of classes in the quotient graph, 
is the number of \emph{connected components}. A graph $G$ is \emph{connected} if it has precisely one component. 
A deletion of link $e=\{x,y\} \in E$ in graph $G=(V,E)$ produces the graph $G-e=(V,E-e)$. 
A link contraction produces a graph $G*e = (V^{\prime},E^{\prime})$, 
with vertex-set $V^{\prime} = V- \{x\}$ and edge-set $E^{\prime}= E-\{\{v,x\},v \in V\} \cup \{\{v,y\}: \{v,x\}\in E\}$. 
Observe that a contraction may result in a multi-graph (i.e., a graph with multiple edges incident to fixed nodes). 
A node deletion produces a new graph $G-v = (V-\{v\},E-\{x,v\}: x \in V)$. 
A cut-vertex $v \in V$ in $G=(V,E)$ is a vertex such that $G-v$ has more components that $G$.\\ 

Now, consider a simple graph $G=(V,E)$, a subset $K \subseteq V$ and a positive integer $d$. 
A subgraph $G_x=(V,E_x) \subseteq G$ is \emph{$d$-$K$-connected} if $d_x(u,v)\leq d, \forall \{u,v\} \subseteq K$, 
where $d_x(u,v)$ is the distance between nodes $u$ and $v$ in the graph $G_x$. 
Let us choose an arbitrary order of the edge-set $E=\{e_1,\ldots,e_m\}$, $e_i\leq e_{i+1}$. 
For each subgraph $G_x=(V,E_x)$ with $E_x \subseteq E$, we identify a binary word $x \in \{0,1\}^m$, 
where $x_i=1$ if and only if $e_i \in E_x$; this is clearly a bijection. 
Therefore, we define the structure $\phi:\{0,1\}^m \to \{0,1\}$ such that $\phi(x)=1$ if $G_x$ 
is $d$-$K$-connected, and $\phi(x)=0$ otherwise. If we assume nodes are perfect but links fail stochastically and independently ruled 
by the vector $p=(p_1,\ldots,p_m)$, the pair $(\phi,p)$ is a coherent SBS. Its reliability, 
denoted by $R_{K,G}^{d}(p)$, is called \emph{diameter constrained reliability}, or DCR for short. 
A particular case is $R_{K,G}^{n-1}(p)$, called \emph{classical reliability}, or CLR for short.

In a coherent SBS, a \emph{pathset} is a state $x$ such that $\phi(x)=1$. 
A \emph{minpath} is a state $x$ such that $\phi(x)=1$ but $\phi(y)=0$ for all $y<x$ (i.e. a minimal pathset). 
A \emph{cutset} is a state $x$ such that $\phi(x)=0$, while a \emph{mincut} is a state $x$ such that $\phi(x)=0$ 
but $\phi(y)=1$ if $y>x$ (i.e. a minimal cutset). We will denote $\mathcal{O}_{d}^{K}(G)$ to the set of all 
$d$-$K$-connected subgraphs of a ground graph $G$. 

\section{Computational Complexity}\label{Prior}
The class $\mathcal{N}\mathcal{P}$ is the set of problems polynomially solvable by 
a non-deterministic Turing machine~\cite{Garey:1979:CIG:578533}. 
A problem is $\mathcal{N}\mathcal{P}$-Hard if it is at least as hard as every problem in the set $\mathcal{N}\mathcal{P}$ 
(formally, if every problem in $\mathcal{N}\mathcal{P}$ has a polynomial reduction to the former). 
It is widely believed that $\mathcal{N}\mathcal{P}$-Hard problems are intractable (i.e. there is no polynomial-time algorithm to solve them). 
An $\mathcal{N}\mathcal{P}$-Hard problem is $\mathcal{N}\mathcal{P}$-Complete if it is inside the class $\mathcal{N}\mathcal{P}$. 
Stephen Cook proved that the joint satisfiability of an input set of clauses in disjunctive form is an $\mathcal{N}\mathcal{P}$-Complete 
decision problem; in fact, the first known problem of this class~\cite{Cook1971}. 
In this way, he provided a systematic procedure to prove that a certain problem 
is $\mathcal{N}\mathcal{P}$-Complete. Specifically, it suffices to prove that the problem is inside the class $\mathcal{N}\mathcal{P}$, 
and that it is at least as hard as an $\mathcal{N}\mathcal{P}$-Complete problem. 
Richard Karp followed this hint, and presented the first 21 combinatorial problems inside this class~\cite{Karp1972}. 
Leslie Valiant defines the class \#$\mathcal{P}$ of counting problems, such that testing whether an element 
should be counted or not can be accomplished in polynomial time~\cite{Valiant1979}. 
A problem is~\#$\mathcal{P}$-Complete if it is in the set \#-$\mathcal{P}$ and it is at least as hard as any problem of that class.  

Recognition and counting minimum cardinality mincuts/minpaths are at least as hard as computing the reliability of a coherent 
SBS~\cite{Ball1986}. Arnon Rosenthal proved the CLR is $\mathcal{N}\mathcal{P}$-Hard~\cite{Rosenthal},  
showing that the minimum cardinality mincut recognition is precisely Steiner-Tree problem, included in Richard Karp's list. 
The CLR for both two-terminal and all-terminal cases are still $\mathcal{N}\mathcal{P}$-Hard, as Michael Ball and J. Scott Provan proved 
by reduction to counting minimum cardinality $s-t$ cuts~\cite{provan83}. 
As a consequence, the general DCR is $\mathcal{N}\mathcal{P}$-Hard as well. Later effort has been focused to particular cases of the DCR, 
in terms of the number of terminals $k=|K|$ and diameter $d$.

When $d=1$ all terminals must have a direct link, $R_{K,G}^{1}= \prod_{\{u,v\} \subseteq K} p(uv)$, 
where $p(uv)$ denotes the probability of operation of link $\{u,v\} \in E$, and $p(uv)=0$ if $\{u,v\} \notin E$. 
Furthermore, the reliability of critical graphs (in the sense that the diameter is increased under any link deletion) is 
again the product of all its elementary reliabilities. An example is the complete graph $K_n$ for diameter $d=1$, 
or a complete bipartite graph when $d=2$.\\ 

The problem is still simple when $k=d=2$. In fact, $R_{\{u,v\},G}^{2} =1- (1-p(uv))\prod_{w \in V-\{u,v\}}(1-p(uw)p(wv))$. 
H\'ector Cancela and Louis Petingi rigorously proved that the DCR is $\mathcal{N}\mathcal{P}$-Hard when $d\geq 3$ and $k\geq 2$ 
is a fixed input parameter, in strong contrast with the case $d=k=2$~\cite{CP2004}. 
The literature offers at least two proofs that the DCR has a polynomial-time 
algorithm when $d=2$ and $k$ is a fixed input parameter~\cite{sartor:thesis,SartorITOR2013a}. 
Pablo Sartor et. al. present a recursive proof~\cite{sartor:thesis}, while Eduardo Canale et. al. present 
an explicit expression for $R_{K,G}^{2}$ that is computed in a polynomial time of elementary operations~\cite{SartorITOR2013a}. 
E. Canale and P. Romero proved the remaining cases, showing that DCR computation belongs to the class of $\mathcal{N}\mathcal{P}$-Hard problems 
in an all-terminal scenario with $d\geq 2$. 

The whole picture of complexity analysis is shown in Fig.~\ref{DCR-complex} as a function of the different pairs $(k,d)$.
\begin{figure}[h!]\centering{
\begin{tikzpicture} [scale=1.6,nodop/.style={inner sep=0pt
}]

\path[draw] (0,1) -- (0,5) -- (5,5) -- (5,1);
\draw[dotted] (-0.3,2) -- (7,2);
\draw[dotted] (-0.3,4) -- (7,4);
\draw[dotted] (1.5,5.3) -- (1.5,4);

\node at (2.5, 5.7) {$k$ (fixed)};
\node at (-0.8, 3) {$d$};
\path[->][dotted] (-0.8, 3.2) edge (-0.8, 4);
\path[->][dotted] (-0.8, 2.8) edge (-0.8, 2);

\node at (0.75,5.3) {$2$};
\node at (1.9, 5.3) {$3\ldots$};
\node at (-0.4, 4.5) {$2$};
\node at (-0.4, 3.8) {$3$};
\node at (-0.4, 3) {.}; \node at (-0.4, 3.1) {.}; \node at (-0.4, 2.9) {.}; 
\node at (-0.4, 2.2) {$n-2$};
\node at (-0.4, 1.8) {$n-1$};
\node at (-0.4, 1.3) {.}; \node at (-0.4, 1.4) {.}; \node at (-0.4, 1.2) {.}; 

\node at (0.75, 4.5) {$O(n)$~\cite{CP2004}};
\node at (3.25, 4.5) {$O(n)$~\cite{SartorITOR2013a}};
\node at (2.5, 3) {$\mathcal{N}\mathcal{P}$-Hard~\cite{CP2004}};
\node at (2.5, 1.5) {$\mathcal{N}\mathcal{P}$-Hard~\cite{Rosenthal}};

\node at (6, 1.5) {$\mathcal{N}\mathcal{P}$-Hard~\cite{provan83}};

\path[draw] (5,5) -- (7,5) -- (7,1);
\path[draw] (5,5) -- (5,6);
\node at (6, 5.7) {$k=n$ or free};
\node at (6, 4.5) {$\mathcal{N}\mathcal{P}$-Hard;~\cite{CanaleRomero}};
\node at (6, 2.50) {$\mathcal{N}\mathcal{P}$-Hard;~\cite{CanaleRomero}};

\end{tikzpicture}} \caption{DCR Complexity in terms of the diameter $d$ and number of terminals $k=|K|$}\label{DCR-complex}
\end{figure}

\section{Exact Computation in Special Graphs}\label{Exact}
So far, an efficient (i.e., polynomial-time) computation of the DCR is available only for special graphs. 
Moreover, all of them tend to be sparse (i.e., with few links). 
To the best of our knowledge, the whole list includes grids~\cite{petingi2013diameter}, weak graphs~\cite{Canale2014134}, 
ladders and spanish fans~\cite{sartor:thesis}. 

A spanish fan is a wheel without a link of its external cycle. Ladders are recursively defined in the thesis~\cite{sartor:thesis}. 
The source-terminal DCR computation for both families of graphs is obtained recursively therein. 
A graph is weak if the removal of an arbitrary subset $E^{\prime}\subseteq E$ of $r$ links 
looses the diameter-connectivity property (specifically, there is a couple of terminals $x,y \in K$ such that $d(x,y)>d$ in the graph $(V,E-E^{\prime})$). 
Weak graphs with $r$ independent of $n$ admit an efficient DCR computation; 
the key idea of the proof is that only graphs with more than $m-r$ links should be tested 
to respect the diameter-constrained property. It is worth to remark that trees, cycles and Monma graphs are weak graphs. 
Since Monma graphs play a central role in robust network design, weak graphs provides a bridge between both areas, 
to know, network reliability and topological network design.  
Louis Petingi added grids to the list, which are planar graphs where each region is a square. 
In its proof, grids are first reduced, using the deletion of irrelevant links, and then finds an exact expression.
The reader can appreciate from Figure~\ref{DCR-complex} that 
an efficient computation is also feasible for diameter $d=2$ and a fixed number of input terminals $k$~\cite{SartorITOR2013a}. 
An explicit expression for $R_{K,G}^{d}(p)$ is provided in~\cite{SartorITOR2013a}.\\ 

We will show that the list of graphs that have exact source-terminal DCR computation will not include all bipartite graphs, 
unless $\mathcal{P}=\mathcal{NP}$. 
First, we name a specific family of graphs, which served as a point of departure in the analysis of source-terminal DCR.
\begin{defn}
Given a bipartite graph $G=(V,E)$ with $V=A \cup B$ and $E \subseteq A \times B$ and a diameter $d\geq 3$. 
Consider an elementary path $P$ with ordered node set $V(P)=\{s,s_1,\ldots,s_{d-3}\}$.  
\emph{Cancela-Petingi network} is defined as the graph $G^{\prime}=(V^{\prime},E^{\prime})$ 
with node-set $V^{\prime} = A \cup B \cup V(P)$ and edge-set $E^{\prime}=E \cup E(P) \cup I$, 
being $I=\{ \{s_{d^{\prime}},a\}, a\in A\} \cup \{ \{b,t \}, b\in B\}$, and 
all links in $G^{\prime}$ are perfect but links in $I$, which fail independently with identical probabilities $p=1/2$. 
\end{defn}

Figure~\ref{fig:covercut} illustrates Cancela-Petingi network for bipartite graph $C_{6}$, when $d=6$. 

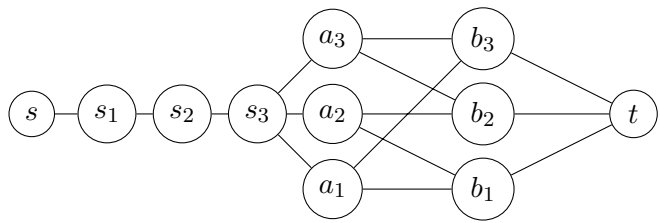
\begin{figure}[h!]\centering{
\begin{tikzpicture}[scale=1]
\tikzstyle{every node}=[draw,shape=circle]
\path 
(-3,0) node (p0) {$s$}
(-2,0) node (p1) {$s_1$}
(-1,0) node (p2) {$s_2$}
 (0,0) node (p3) {$s_3$}

(1,-1) node (p4) {$a_1$}
(1, 0) node (p5) {$a_2$}
(1, 1) node (p6) {$a_3$}

(3,-1) node (p7) {$b_1$}
(3, 0) node (p8) {$b_2$}
(3, 1) node (p9) {$b_3$}

(5,0) node (p10) {$t$};

\draw 
(p0) -- (p1)
(p1) -- (p2)
(p2) -- (p3)

(p4) -- (p7)
(p4) -- (p9)

(p5) -- (p7)
(p5) -- (p8)

(p6) -- (p8)
(p6) -- (p9)

(p3) -- (p4)
(p3) -- (p5)
(p3) -- (p6)
(p7) -- (p10)
(p8) -- (p10)
(p9) -- (p10);
\end{tikzpicture}} \caption{Cancela-Petingi network with terminal set $\{s,t\}$ and $d=6$, 
for the particular bipartite instance $C_6$.}\label{fig:covercut}
\end{figure}

Cancela and Petingi used this family of graphs in order to prove that the source-terminal DCR belongs to the class of 
$\mathcal{NP}$-Hard problems~\cite{CP2004}. They showed that the reliability computation of these networks of these graphs is at least 
as hard as counting vertex covers in bipartite graphs, which is indeed a \#$\mathcal{P}$-Complete problem~\cite{provan83}. 

\begin{cor} \label{bipartite}
The general DCR computation in bipartite graphs belongs to the class of $\mathcal{NP}$-Hard problems.
\end{cor}
\begin{proof}
Bipartite graphs include Cancela-Petingi graphs. The DCR computation in Cancela-Petingi graphs is at least as hard as counting 
vertex covers in bipartite graphs.
\end{proof}

Curiously, all graphs from the list are sparse. Proposition~\ref{bipartite} confirms that it is hard to find the DCR of dense bipartite graphs. 
The general DCR computation in complete graphs is $\mathcal{NP}$-Hard as well 
(all graphs are particular cases of the complete graph, choosing $p(uv)=0$ for some links).\\

From now on, we will stick to the source terminal case, where the terminal-set will be named $K=\{s,t\}$. 
The DCR for a new family of graphs will be efficiently found for this case.

\begin{defn}
Let $H$ and $G$ be two graphs. Let $e=xy$ an distinguished edge of $H$ and $s,t$ distinguished in $G$. 
We define the \emph{replacement} of $e$ by $G$ in $H$, and write $H+_eG$ to the graph obtained from $H-e\cup G$ 
by identifying $x$ with $s$ and $y$ with $t$. 
\end{defn}
For example, if $H$ is a triangle with nodes $\{1,2,3\}$, we consider $e=\{1,2\}$ and $G$ a triangle with nodes $\{s,t,u\}$,  
then the replacement of $e$ by $G$ in $H$ is the new graph $H+_{e}G$ with node-set $\{1,2,3,s,t,u\}$ and 
edge-set $\{13,23,su,ut,st\}$. 

\begin{defn}
Let $H=(V,E)$ and $G$ be two graphs, where $s,t$ represent distinguished nodes in $G$. 
The \emph{replacement} of $G$ in $H$ is denoted by $H_G$, and is obtained replacing successively all links from $H$ by $G$. 
If $E = \{ \{x_1,y_1\},\ldots, \{x_m,y_m\}\}$ then: 
\begin{equation}
H_G = H+(x_1,y_1)G+(x_2,y_2)G+\ldots+(x_m,y_m)G, 
\end{equation} 
\end{defn}

In words, we replace each link of $H$ by the source-terminal graph $G$.
\begin{prop}\label{generalizada}
Consider any fixed graph $H$, $K= \{u,v\} \subseteq V(H)$, $\{x,y\} \in E(H)$, and a positive integer $d\geq 1$. 
Then, $R_{K, H+_eG}^{d}$ depends polynomially on $R_{\{s,t\},G}^{d}$ for any fixed $K=\{u,v\}$ in $V(H)$. 
In particular, the same result holds for $H_G$.
\end{prop}
\begin{proof}
Minpaths of a source-terminal scenario are $s-t$ paths.  
Tese minpaths in $H$ could consist in a large number of paths, but constant in $H$, since $H$ is fixed. 
Now, the minpaths of $H_G$ consists on paths of the form 
$$
v_0^1, \ldots, v_{n_1}^1 = v_0^2, \ldots, v_{n_1}^2 =v_{0}^3,\ldots,v_{n_{k-1}}^{k-1}= v_0^k, \ldots, v_{n_k}^k, 
$$
where 
$$
v_0^1, v_0^2, \ldots, v_{0}^{k}, v_{n_k}^k 
$$ 
is a minpath in $H$, and for  $j = 1, \ldots, k$
$$
v_0^j, \ldots, v_{n_1}^j,
$$
is a minpath in $G$. 
Let us notice that the equality $v_{n_j}^j = v_0^{j+1}$  is in $H_G$, but, 
though $v_{n_j}^j$ and $v_0^{j+1}$ come from the vertices $\{s,t\}$, they could come from different vertices,  for instance, 
$v_{n_j}^j$ could come from $s$ but $v_0^{j+1}$ from $t$. 
By Inclusion-Exclusion formula, the reliability of $H$ is a polynomial on the reliability of its minpaths. 
If $R_i$ from $i=1,\ldots, L$ represent the reliability of the $i$-th minpaths of $H$ in some order, then, there exits a multinomial 
$P(x_1,x_2,\ldots,x_L)$ such that 
\begin{equation}\label{branches}
R_{\{u,v\},H}^{d} = P(R_1,R_2,\ldots, R_L).
\end{equation}

Then, if the $i$-th minpath has length $n_i$, we have
\begin{equation}\label{sum}
R_i = \sum_{d_1+\ldots+d_{n_i} \leq d} \prod_{j=1}^{n_j} R_{\{s,t\},G}^{d_j}
\end{equation} 
\end{proof}

Proposition~\ref{generalizada} tells us that the replacement operation is a complexity invariant.
%

A closed formula for $R_{\{s,t\},H_G}^{d}$ is obtained replacing~\eqref{branches} into \eqref{sum}: 
\begin{equation}\label{expression}
R_{\{u,v\},H}^{d} = P(\sum_{d_1+\ldots+d_{n_i} \leq d} \prod_{j=1}^{n_j} R_{\{s,t\},G}^{d_j},\ldots, \sum_{d_1+\ldots+d_{n_l} \leq d} \prod_{j=1}^{n_l} R_{\{s,t\},G}^{d_j}).
\end{equation}
Expression~\eqref{expression} could be useful in order to compare the performance of exact approaches by replacement in fixed graphs. 
We recommend to pick a member $G$ from the list of graphs that accept polynomial-time DCR computation.

\section{Exact computation and Irrelevant Components}\label{Irrelevant-Components}
In this section, we explore exact computation of the diameter constrained reliability for general graphs. 
A natural approach is to enrich this field with relevant results coming from classical reliability 
(here, we will term \emph{classical reliability} as reliability with no diameter-constraint). 

The literature that deals with exact computation in the classical reliability measure is vast and rich. 
We invite the reader to find a list of references in the survey~\cite{NET:NET3230250308}, and Chapter 7 from~\cite{rubino2009rare}. 
In~\cite{NET:NET3230250308}, Suresh Rai et. al. suggest a classification of exact approaches into three classes, to know:
\begin{itemize}
 \item Inclusion-Exclusion methods.
 \item Sum of Disjoint Products (SDPs). 
 \item Factorization methods.

\end{itemize}
Inclusion-Exclusion methods are still valid for arbitrary stochastic binary systems. Indeed, the key idea is to enumerate all minpaths 
$M_1,\ldots,M_r$ and to compute the reliability $R$ in the following manner:
\begin{equation} \label{uniones}
R = P(\cup_{i=1}^{r} M_i), 
\end{equation}
where the last union is developed using inclusion-exclusion principle. The main drawback of this family of techniques is that 
they deal with a full enumeration of minpaths (or cutsets). 

The key idea from the second class of methods is to re-write Expression~\eqref{uniones} in a sum of disjoint products, in order to 
translate that probability into a sum. Valuable references are available in the classical reliability context~\cite{NET:NET3230250308}.  
To the best of our knowledge, there is no development of SDPs to find the diameter-constrained reliability. 
For instance, Ahmad method finds all branches from source to terminal, and the sums their probabilities. 
Ahmad method~\cite{Ahmad1982} can be adapted in a straightforward manner, 
by means of a deletion of branches which are larger than the diameter $d$. 
However, a more efficient method would delete those branches online, while branches are built.

A seminal work from the third class of methods is due to Moskowitz, inspired in electrical circuits~\cite{6372698}. 
He observed that in the $K$-terminal reliability, the following contraction-deletion formula holds:
\begin{equation}
R_{K,G} = (1-p_e)R_{K,G-e}+pR_{K^{\prime},G*e}, 
\end{equation}
being $G*e$ the contraction of link $e=\{x,y\}$, and $K^{\prime}$ the new terminal set in $G*e$, where the identified node 
is a terminal if and only if $x$ or $y$ are terminals in $G$. This idea helps to successively reduce the graph, and 
finish the sequence in a trivial graph (non-connected or trivially connected). The process requires an exponential 
number of operations in a worst case, but it can be accelerated choosing the components to delete in a specific order. 
In network reliability context, it is recommended to combine deletion-contraction formula with the deletion of irrelevant 
components iteratively, in order to accelerate the reduction of the network size~\cite{NET:NET3230250308,doi:10.1137/0214057}.

The determination of irrelevant components (links) in classical network reliability is already understood. 
However, finding irrelevant links in diameter constrained reliability is a new topic, that has been addressed in recent works, 
for the specific source-terminal case~\cite{CEKP2012,petingi2013diameter}. 
Irrelevance in the diameter-constraint scenario is formalized in the following decision problem.
\begin{defn}[Relevance of a Link in DCR]
INPUT: graph $G=(V,E)$, nodes $s,t \in V$, link $e \in E$ and positive integer $d: 2\leq d \leq |V|$.\\
PROPERTY: $G$ has an $s-t$ path $P$ with length not greater than $d$, such that $e\in P$. 
\end{defn}

Suppose that we require to connect the terminals $K=\{st\}$ 
in a network $G$ using $d$ hops or less. Therefore, a link $e=\{x,y\}$ is \emph{irrelevant} if there is no elementary path $P$ 
with length $l(P)\leq d$ such that $e \in P$. The first work in the area is provided by Cancela et. al.~\cite{CEKP2012}. 
The authors propose a sufficient condition to determine when a link is irrelevant. Specifically, they observed that 
if $d(s,x)+d(y,t)\geq d$ and $d(s,y)+d(x,t)\geq d$, then $e=\{x,y\}$ is irrelevant. 
However, they pointed out that the idea does not work for the graph from 
Figure~\label{fig:red}. There, $d(s,x)=1$ and $d(y,t)=2$, but link $e=\{x,y\}$ is irrelevant when $d=6$. 
Louis Petingi proposed a stronger sufficient condition: if $d_{G-e}(s,x)+d_{G-e}(y,t)\geq d$ and $d_{G-e}(s,y)+d_{G-e}(x,t)\geq d$, 
then $e=\{x,y\}$ is irrelevant. Observe the condition is identical, but distances are taken in the graph $G-e$, 
since link $e$ should be used only once. Let us call Conditions 1 and 2, in chronological order. 
In Figure~, we can appreciate link $e$ is irrelevant whenever $d \leq 6$. In order to fix ideas, consider $d=5$. 
Condition $1$ does not return link $e$ as irrelevant, since $d(s,x)+d(y,t) = 1+2< 5$. However, Condition $2$ certificates correctly 
that $e$ is irrelevant, since $d_{G-e}(s,x)+d_{G-e}(y,t)=1+4=5<6$ and $d_{G-e}(s,y)+d_{G-e}(x,t)=4+1=5 \geq 5$.\\

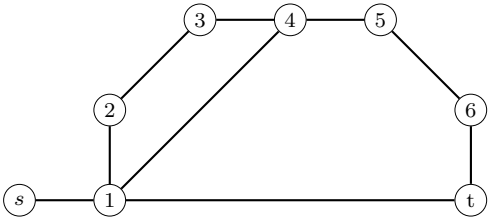
\begin{figure}[h!]\centering{
	 \begin{tikzpicture}  
	 [scale=1.2,>=stealth,shorten >=0.1pt, auto,  thick,
	nodot/.style={circle, draw = black , very thin, minimum size=12pt, inner sep=0pt, font=\scriptsize },                     
	 ]
	\begin{scope}[xshift=0cm,scale=1]
	\node [nodot] (s) at (0,0) {$s$};
	\node [nodot] (1) at (1,0) {1};
	\node [nodot] (2) at (1,1) {2};
	\node [nodot] (3) at (2,2) {3};
	\node [nodot] (4) at (3,2) {4};
	\node [nodot] (5) at (4,2) {5};
	\node [nodot] (6) at (5,1) {6};
	\node [nodot] (t) at (5,0) {t};

      \draw (s)--(1)--(2)--(3)--(4)--(5)--(6)--(t);
      \draw (1)--(4);      
      \draw (1)--(t);   
	\end{scope}
	\end{tikzpicture}} \caption{Sample graph $G$ with an irrelevant link $e=\{1,2\}$ when $d=6$.}
\label{fig:red}
\end{figure}

We invite the reader to check that $e$ is irrelevant when $d=6$, but Conditions 1 and 2 fail to give a positive certificate of irrelevance. 
In this paper, we will present a stronger sufficient condition (that we will term Condition 3 for chronological reasons). Let us denote 
$sPt$ an $s-t$ path, and the concatenation of paths $P_1$ and $P_2$ as $P_1P_2$. The directed path composed by link $\{x,y\}$ from $x$ to $y$ 
is denoted $(x,y)$. The key point is that a successful path $sPt$ should either present the shape $sP_1(xy)P_2t$ or $sP_1(yx)P_2t$, 
In both cases, paths $P_1$ does not contain nodes $t$ neither $y$, and $P_2$ does not contain $x$ neither $s$ must be disjoint. 

\begin{prop}
Let us consider a simple graph $G=(V,E)$, two terminal nodes $\{s,t\} \in K \subseteq V$ and a link $e= \{x,y\} \in E$. 
If $d_{G-y-t}(s,x) + d_{G-s-x}(y,t)\geq d$ and $d_{G-x-t}(s,y) + d_{G-s-y}(x,t)\geq d$, then $e$ is irrelevant. 
\end{prop}
\begin{proof}
Suppose $P$ is a successful path (i.e., $e\i P$ and $l(P)\leq d$). There are two possible cases:
\begin{itemize}
 \item $P = sP_1(xy)P_2t$ where $P_1 \subseteq G-y-t$ and $P_2 \subseteq G-x-s$. Then 
$l(P)= l(P_1)+1+l(P_2) \geq d_{G-y-t}(s,x)+1+d_{G-x-s} \geq d+1$, which is a contradiction, since $l(P)\leq d$.
  \item $P = sP_1(yx)P_2t$ where $P_1 \subseteq G-t-x$ and $P_2 \subseteq G-y-s$. Then 
$l(P)= l(P_1)+1+l(P_2) \geq d_{G-x-t}(s,x)+1+d_{G-y-s} \geq d+1$, which is a contradiction, since $l(P)\leq d$.
\end{itemize}
Therefore, there is no successful path, and link $e$ is irrelevant.
\end{proof}

The graph inclusions $G-e \subseteq G-y-t$ and $G-e \subseteq G-x-s$ hold in general. 

Then, $d_{G-e}(s,x)+d_{G-e}(y,t) \leq d_{G-y-t}(s,x) + d_{G-s-x}(y,t)$. As a consequence, if $d_{G-e}(s,x)+d_{G-e}(y,t) \geq d$ 
then $d_{G-y-t}(s,x) + d_{G-s-x}(y,t)\leq d$, and Condition 3 is stronger than Condition 2. 

Let us return to Figure~\ref{fig:red} for diameter $d=6$. Condition 2 fails, since $d_{G-e}(s,x)+d_{G-e}(y,t)=1+4=5<6$. 
However, Condition 3 gives a positive certificate of irrelevance, since 
$d_{G-y-t}(s,x) + d_{G-s-x}(y,t) = 1 + 5 = 6 \geq 6$, and $d_{G-x-t}(s,y) + d_{G-s-y}(x,t) = \infty \geq d$.  
In this case it is impossible to visit node $y$ before $x$, and $t$ is unreachable from $x$ in graph $G-x-t$.\\

We have already three sufficient conditions to determine whether a certain link is irrelevant or not. In~\cite{CEKP2012}, the authors 
state that an efficient condition to determine irrelevant links is still an open problem. 
A classical result from complexity theory is that finding the Longest Path between two nodes is an $\mathcal{NP}$-Hard problem 
(since Hamilton Path can be reduced to it~\cite{Garey:1979:CIG:578533}). 
More recently, Louis Petingi proved that finding the longest path through a certain link 
belongs to the class of $\mathcal{NP}$-Hard problems, where links can assume integer costs~\cite{petingi2013diameter}. 
Let us show here that Relevance of a Link in DCR is a particular subset of instance of the following decision problem:
\begin{defn}[Shortest Path]
INPUT: Graph $G=(V,E)$, costs $c:E\to \mathbb{Z}$ in the edges, nodes $s,t \in V$, integer $k$.\\
PROPERTY: $G$ has an $s-t$ path $P$ with length not greater than $k$. 
\end{defn}
Longest Path can be reduced to Shortest Path taking opposite values for edges. 
Moreover, a positive certificate of an arbitrary instance for Shortest Path is directly recognizable. 
Then, Shortest Path belongs to the class of $\mathcal{NP}$-Complete problems.\\

Let us show here that Relevance of a Link in DCR is included in Shortest Path. First, we warn that this is not a proof of hardness 
(just to know, it is a special sub-problem of a hard one). Consider an arbitrary instance $(G,s,t,e,d)$ for Relevance of a Link in DCR. 
Assign unit cost $c(e^{\prime})=1$ to all links $e^{\prime} \in E-\{e\}$, but $c(e)=-(d-1)$. 
Consider Shortest Path instance $(G,s,t,c,0)$, this is, with cost function $c$ and integer $k=0$. 
If we are able to find a path $P$ with cost $c(P)\leq 0$, 
the path necessarily contains edge $e$. 
Additionally, if it has length $l(P)=l$ then its cost equals $c(P)=(l-1)-(d-1)=l-d \leq 0$, so $l\leq d$. In that case, path $P$ is a solution 
for Relevance of a Link in DCR. 

Here, we just elucidate a bridge between Shortest Path an our problem of interest. However, an efficient determination  
of irrelevant links in DCR is still an open problem.

\section{Algorithm} \label{Algorithm}
We will combine the strength of Moskowitz decomposition formula (that is adapted in a diameter constrained scenario) with 
Condition 3 and elementary operation that are DCR invariants. The spirit of this technique is either simplify 
or reduce the size of the graph as most as possible during each iteration.  
We will stick to the previous notation (graph $G$, source and terminals $s,t$, diameter $d$):
\begin{itemize}
 \item \emph{Pending-Node}: If the source $s$ (idem terminal $t$) is pending on a link $e=\{s,x\}$ with reliability $p_e$, then we contract link $e$, and 
replace $G$ for its contraction $G*e$. The invariant is $R_{\{s,t\},G}^{d} = R_{\{s^{\prime},t\},G*e}^{d-1}$, being $e^{\prime}$ the 
new source. All non-terminal pending nodes are deleted.
 \item \emph{Perfect-Path}: If a path $P=\{v_1,\ldots,v_n\}$ is an induced subgraph for $G$ and links have elementary reliabilities $p(v_i,v_i+1)$. 
then we re-assign the link reliabilities $p(v_i,v_{i+1})=1$ for all $i=1,\ldots,n-2$ but $p(v_{n-1},v_n)= \prod_{i=1}^{n-1}p(v_i,v_{i+1})$.
 \item \emph{Perfect-Neighbors}: if the source $s$ (idem terminal $t$) has all perfect links to its neighbors $N(s)$, then $s \cup N(s)$ is a 
new vertex in $G^{\prime}$ and $R_{\{s,t\},G}^{d} = R_{\{s,t\},G^{\prime}}^{d-1}$.
 \item \emph{Perfect-Cut-Node}: if $v$ is a cut-node (i.e. $G-v$ has more than one component), first delete 
components with all non-terminal nodes (observe that we cannot finish with more than two components). 
Second, apply Perfect-Neighbors to $v$ on both sides.
 \item \emph{Parallel-Links}: If we find two links $e_1$ and $e_2$ from the same nodes with elementary reliabilities $p_1$ and $p_2$, 
they are replaced by a new link $e$ with reliability $p_1+p_2-p_1p_2$. 
\end{itemize}
This set of five elementary operations are combined with Condition 3 and Moskowitz decomposition in Algorithm $IP_5M$. 
The name $IP_5M$ stands for the acronym $I$ (deletion of Irrelevant links), $P_5$ (apply the ``5P'' elementary operations), and 
$M$ (Moskowitz decomposition).

\begin{algorithm}[H]
\caption{$R = IP_5M(G,s,t,d)$} \label{BLocal}
\begin{algorithmic}[1]
\IF{$HasPefectPath(G,d)$}
\RETURN $R=1$
\ENDIF
\IF{$Distance(s,t)<d$}
\RETURN $R=0$
\ENDIF
\STATE $G \leftarrow I(G,Condition3)$
\STATE $(G,s,t,d) \leftarrow 5P(G,s,t,d)$
\STATE $e \leftarrow NonPerfectRandom(E(G))$
\RETURN $(1-p_e) IP_5M(G-e,s,t,d)+p_e IP_5M(G*e,s,t,d)$
\end{algorithmic}
\end{algorithm}

The block of Lines $1$ to $6$ tests the termination (either we will have a perfect path of length $d$ or less, or there is no $s-t$ path 
with such distance). The extreme situations are guaranteed since links will take values in $\{0,1\}$ during the process. In Line $7$, 
Condition $3$ is revised in every link. Each irrelevant link is there removed. Line $8$ deserves further explanation. 
There, the five ``P'' operations (to know, Pending-Terminal, Perfect-Path, Perfect Neighbors, 
Perfect Cut-Node and Parallel-Links) are introduced, following the list order. Here we remark that each elementary operation takes place 
until no additional reduction with that operation is possible. The reader is invited to check that all elementary operations 
run in linear time, but Perfect-Cut-Node, where it could run in order $O(|V|^2)$ if the number of cut nodes is $O(|V|)$ 
(in that case, this elementary operation will reduce the graph in a large manner). 

For a better understanding of Algorithm $IP_5M$, we will introduce network $G$ from Figure~\ref{fig:red} with diameter $d=6$ and identical 
probabilities $p_e=p \in (0,1)$ as a pictorial example. Since $s-1-t$ is a path with length two, Block of lines $1$-$6$ do not 
meet a termination case. In Line $7$, we test links for Condition $3$. As seen before, link $\{1,2\}$ is correctly detected as irrelevant. 
The reader can observe that links $\{2,3\}$ and $\{3,4\}$ are irrelevant as well, but they are not detected by Condition $3$. 
Then, pending nodes are found. Nodes $2$ and $3$ are deleted, while a contraction of nodes $s$ and $1$ takes place, and 
$R_{\{s,t\},G}^{6} = pR_{\{s,t\},C_5}^{5}$. Observe that $G$ has already been reduced to a cycle. Then, Perfect-Path is 
applied, and links $\{1,4\},\{4,5\}$ and $\{5,6\}$ are perfect, while $p(5,6)=p^4$. Immediately, Perfect-Neighbors is called, 
with no effect in this case. The same occurs with Functions Cut-Node and Parallel-Links. In Line $9$, a non-perfect link is chosen 
uniformly at random. In the case, links $\{6,t\}$ and $\{1,t\}$ are the only possibilities. Without loss of generality, let us assume 
that $\{6,t\}$ is chosen. Correspondingly, Moskowitz decomposition formula is applied using link $e=\{6,t\}$, and:
\begin{equation}
R_{\{s,t\},C_5}^{5} = (1-p^4)R_{\{s,t\},P_4}^{5}+p^4R_{\{s,t\},C_4}^{5} 
\end{equation}
Finally, $IP_5M$ is called twice in order to find respectively $R_{\{s,t\},P_4}^{5}$ and $R_{\{s,t\},C_4}^{5}$. 
In the first network $P_4$, pending non-terminal nodes $6$ and $5$ and then $4$ will be removed, and the only path $\{1,t\}$ has probability $p$, 
so $R_{\{s,t\},P_4}^{5}=p$. In the call for the cycle $C_4$, a perfect path between terminals is found, and 
$R_{\{s,t\},C_4}^{5}=1$. The result is precisely:
\begin{equation}\label{correct}
R_{\{s,t\},G}^{6}(p) =  pR_{\{s,t\},C_5}^{5} = p [(1-p^4)p + p^4]=(1-p^4)p^2+p^{5}.
\end{equation}
By a sum of disjoint products, a direct verification is available. 
 Let us consider Figure~\ref{fig:red} again. Links $\{1,2\},\{2,3\}$ and $\{3,4\}$ are irrelevant. The success 
occurs either when path $s-1-4-5-6-t$ works (probability $p^5$) or it does not but path $s-1-t$ works (event with probability $(1-p^4)p^2$). 
Both events are mutually exhaustive and disjoints for the connectivity success under diameter constrained $d=6$; so, these terms 
must be added, and Equation~\eqref{correct} is retrieved.

\section{Concluding Remarks and Trends for Future Work} \label{Conclusion}
The system under study has perfect nodes but imperfect links, that fail stochastically and independently. 
The classical reliability problem aims to find the connectedness probability of target nodes, called terminals. 
In our system, we also require terminals to be within a specified number of hops or less, called diameter.\\

The theory of diameter-constrained reliability DCR inherits several ideas from classical reliability. 
Its hardness is inherited; hence exponential time exact algorithms have been presented so far in order to find de DCR of general graphs. 
In the $K$-terminal reliability, an efficient computation of the DCR is available in specific families of graphs, 
to know, weak graphs, ladders and spanish fans. In the source-terminal scenario, a multinomial recursive expression for the 
DCR is presented for complete graphs with four different link classes (i.e., different elementary reliabilities on their link), 
and grid graphs (i.e., planar graphs such that every region is a square) also accept efficient DCR computation.  

Additionally, Monte-Carlo based approaches and other approximation techniques have been inherited from classical reliability as well, 
with minor modifications. The three classes of methods for exact computation (Sum of Disjoint Product, Inclusion-Exclusion and 
Factorization) require adaptations. Remarkably, the contraction operation is not feasible (since the diameter is modified), 
and Moskowitz rule decomposes the problem in an on-off link state, where ``on'' state now means that the link is perfect. 
A challenging task now is to determine irrelevant links, even in a source-to-terminal context, when the diameter constraint is included.\\

In this work we studied exact DCR computation for source-terminal scenario. 
We introduce a recursive family of replacement graphs, that receive a certain graph $G$ that has an 
exact and efficient DCR computation, and returns the DCR for a replacement $H_G$ of $G$ by each link in $H$. 
Replacement is an complexity invariant.   

Additionally, a new sufficient condition for the determination of relevant links is provided, together with 
elementary operations that are DCR invariants. They were combined with Moskowitz decomposition, resulting in an algorithm, 
called $IP_5M$. This algorithm iteratively reduces/simplifies the graph putting all previous results together. 
A pictorial example shows the strength of this algorithm for graphs with reduced size.\\

There are several aspects that deserve future work. A full comparative analysis between the different exact algorithms 
in the literature will give an insight of the effectiveness and performance (in terms of memory and computational effort). 
An efficient way to determine irrelevant links is still an open problem. 
An exhaustive (and efficient) construction of sum of disjoint products should be explored. 
For instance, Ahmad method shows to be simple and satisfactorily finds a sum of disjoint products, called \emph{branches}, in 
the source-terminal classical reliability scheme. These branches have been adapted to the context of classical $K$-terminal reliability. 
A natural step is to extend Ahmad method for DCR, first in source-terminal and then in $K$-terminal settings. 
Different exact algorithms can be compared on the lights of replacement graphs.

\bibliography{bib-dcr}

\providecommand{\etalchar}[1]{$^{#1}$}
\providecommand{\bysame}{\leavevmode\hbox to3em{\hrulefill}\thinspace}
\providecommand{\noopsort}[1]{}
\providecommand{\mr}[1]{\href{http://www.ams.org/mathscinet-getitem?mr=#1}{MR~%
#1}}
\providecommand{\zbl}[1]{\href{http://www.zentralblatt-math.org/zmath/en/searc%
h/?q=an:#1}{Zbl~#1}}
\providecommand{\jfm}[1]{\href{http://www.emis.de/cgi-bin/JFM-item?#1}{JFM~#1}}
\providecommand{\arxiv}[1]{\href{http://www.arxiv.org/abs/#1}{arXiv~#1}}
\providecommand{\doi}[1]{\url{http://dx.doi.org/#1}}
\providecommand{\MR}{\relax\ifhmode\unskip\space\fi MR }
\providecommand{\MRhref}[2]{%
  \href{http://www.ams.org/mathscinet-getitem?mr=#1}{#2}
}
\providecommand{\href}[2]{#2}
\begin{thebibliography}{CCR{\etalchar{+}}13}

\bibitem[Ahm82]{Ahmad1982}
\bgroup\scshape{}S.~H. Ahmad\egroup{}, A simple technique for computing network
  reliability,  \emph{IEEE Transactions on Reliability} \textbf{R-31} (1982),
  41 --44. \doi{10.1109/TR.1982.5221220}.

\bibitem[Bal86]{Ball1986}
\bgroup\scshape{}M.~O. Ball\egroup{}, Computational complexity of network
  reliability analysis: An overview,  \emph{IEEE Transactions on Reliability}
  \textbf{35} (1986), 230 --239. \doi{10.1109/TR.1986.4335422}.

\bibitem[CR14]{CanaleRomero}
\bgroup\scshape{}E.~Canale\egroup{} and \bgroup\scshape{}P.~Romero\egroup{},
  {Diameter Constrained Reliability: Computational Complexity in terms of the
  diameter and number of terminals},  \emph{Arxiv preprint Computer Science}
  (2014). Available at \url{http://arxiv.org/abs/1404.3684}.

\bibitem[CCR{\etalchar{+}}13]{SartorITOR2013a}
\bgroup\scshape{}E.~Canale\egroup{}, \bgroup\scshape{}H.~Cancela\egroup{},
  \bgroup\scshape{}F.~Robledo\egroup{}, \bgroup\scshape{}G.~Rubino\egroup{},
  and \bgroup\scshape{}P.~Sartor\egroup{}, On computing the
  2-diameter-constrained {K}-reliability of networks,  \emph{International
  Transactions in Operational Research} \textbf{20} (2013), 49--58.

\bibitem[CRRS14]{Canale2014134}
\bgroup\scshape{}E.~Canale\egroup{}, \bgroup\scshape{}F.~Robledo\egroup{},
  \bgroup\scshape{}P.~Romero\egroup{}, and \bgroup\scshape{}P.~Sartor\egroup{},
  Monte carlo methods in diameter-constrained reliability,  \emph{Optical
  Switching and Networking} \textbf{14, Part 2} (2014), 134 -- 148, Special
  Issue on RNDM 2013.

\bibitem[CP04]{CP2004}
\bgroup\scshape{}H.~Cancela\egroup{} and \bgroup\scshape{}L.~Petingi\egroup{},
  Reliability of communication networks with delay constraints: computational
  complexity and complete topologies,  \emph{International Journal of
  Mathematics and Mathematical Sciences} \textbf{2004} (2004), 1551--1562.

\bibitem[CKP11]{CEKP2012}
\bgroup\scshape{}H.~Cancela\egroup{}, \bgroup\scshape{}M.~E. Khadiri\egroup{},
  and \bgroup\scshape{}L.~Petingi\egroup{}, Polynomial-time topological
  reductions that preserve the diameter constrained reliability of a
  communication network.,  \emph{IEEE Transactions on Reliability} \textbf{60}
  (2011), 845--851.

\bibitem[Col99]{Colbourn99reliabilityissues}
\bgroup\scshape{}C.~J. Colbourn\egroup{}, Reliability issues in
  telecommunications network planning,  in \emph{Telecommunications network
  planning, chapter 9}, Kluwer Academic Publishers, 1999, pp.~135--146.

\bibitem[Coo71]{Cook1971}
\bgroup\scshape{}S.~A. Cook\egroup{}, The complexity of theorem-proving
  procedures,  in \emph{Proceedings of the third annual ACM symposium on Theory
  of computing}, \emph{STOC '71}, ACM, New York, NY, USA, 1971, pp.~151--158.

\bibitem[GJ79]{Garey:1979:CIG:578533}
\bgroup\scshape{}M.~R. Garey\egroup{} and \bgroup\scshape{}D.~S.
  Johnson\egroup{}, \emph{Computers and Intractability: A Guide to the Theory
  of NP-Completeness}, W. H. Freeman and Company, New York, NY, USA, 1979.

\bibitem[Kar72]{Karp1972}
\bgroup\scshape{}R.~M. Karp\egroup{}, Reducibility among combinatorial
  problems,  in \emph{Complexity of Computer Computations}
  (\bgroup\scshape{}R.~E. Miller\egroup{} and \bgroup\scshape{}J.~W.
  Thatcher\egroup{}, eds.), Plenum Press, 1972, pp.~85--103.

\bibitem[Mos58]{6372698}
\bgroup\scshape{}F.~Moskowitz\egroup{}, The analysis of redundancy networks,
  \emph{American Institute of Electrical Engineers, Part I: Communication and
  Electronics, Transactions of the} \textbf{77} (1958), 627--632.
  \doi{10.1109/TCE.1958.6372698}.

\bibitem[PR01]{PR01}
\bgroup\scshape{}L.~Petingi\egroup{} and
  \bgroup\scshape{}J.~Rodriguez\egroup{}, Reliability of networks with delay
  constraints,  in \emph{Congressus Numerantium}, \textbf{152}, 2001,
  pp.~117--123.

\bibitem[Pet13]{petingi2013diameter}
\bgroup\scshape{}L.~Petingi\egroup{}, Diameter-related properties of graphs and
  applications to network reliability theory,  \emph{Reliability and
  Vulnerability in Models and Its Applications} \textbf{12} (2013).

\bibitem[PB83]{provan83}
\bgroup\scshape{}S.~J. Provan\egroup{} and \bgroup\scshape{}M.~O.
  Ball\egroup{}, {The Complexity of Counting Cuts and of Computing the
  Probability that a Graph is Connected},  \emph{SIAM Journal on Computing}
  \textbf{12} (1983), 777--788.

\bibitem[RVT95]{NET:NET3230250308}
\bgroup\scshape{}S.~Rai\egroup{}, \bgroup\scshape{}M.~Veeraraghavan\egroup{},
  and \bgroup\scshape{}K.~S. Trivedi\egroup{}, A survey of efficient
  reliability computation using disjoint products approach,  \emph{Networks}
  \textbf{25} (1995), 147--163.

\bibitem[Ros77]{Rosenthal}
\bgroup\scshape{}A.~Rosenthal\egroup{}, Computing the reliability of complex
  networks,  \emph{SIAM Journal on Applied Mathematics} \textbf{32} (1977),
  384--393.

\bibitem[RT09]{rubino2009rare}
\bgroup\scshape{}G.~Rubino\egroup{} and \bgroup\scshape{}B.~Tuffin\egroup{},
  \emph{Rare event simulation using {M}onte {C}arlo methods}, Wiley, 2009.

\bibitem[Sar13]{sartor:thesis}
\bgroup\scshape{}P.~Sartor\egroup{}, \emph{{Propri\'{e}t\'{e}s et m\'{e}thodes
  de calcul de la fiabilit\'{e} diam\`{e}tre-born\'{e}e des r\'{e}seaux}},
  Ph.D. thesis, INRIA/IRISA, Universit\'{e} de Rennes I, Rennes, France,
  december 2013.

\bibitem[SW85]{doi:10.1137/0214057}
\bgroup\scshape{}A.~Satyanarayana\egroup{} and
  \bgroup\scshape{}R.~Wood\egroup{}, A linear-time algorithm for computing
  k-terminal reliability in series-parallel networks,  \emph{SIAM Journal on
  Computing} \textbf{14} (1985), 818--832. \doi{10.1137/0214057}.

\bibitem[Val79]{Valiant1979}
\bgroup\scshape{}L.~Valiant\egroup{}, The complexity of enumeration and
  reliability problems,  \emph{SIAM Journal on Computing} \textbf{8} (1979),
  410--421.

\end{thebibliography}
\bibliographystyle{aomalpha}
\end{document}